\documentclass[a4paper]{article}

\usepackage[margin=3cm]{geometry}

\usepackage{amsmath, amssymb, amsthm, latexsym, dsfont, thmtools}

\usepackage{algorithm}
\usepackage{algpseudocode}

\usepackage{graphicx}
\usepackage{tikz}
\usepackage{tkz-euclide}
\usepackage{pgfplots,todonotes}
\pgfplotsset{compat=1.18}
\usetikzlibrary{fillbetween,backgrounds}

\usepackage[T1]{fontenc}

\usepackage[dvipsnames]{xcolor}

\usepackage{booktabs, makecell, tabularx, colortbl}

\usepackage[colorlinks=true]{hyperref}
\usepackage{cleveref}
\hypersetup{
    linkcolor=blue,   
    citecolor=blue,   
    urlcolor=magenta,  
    filecolor=blue, 
    pdftitle={Problems on Vertex Splitting},
    pdfauthor={Ajinkya Gaikwad, Hitendra Kumar, Soumen Maity, S. Padmapriya, Praneet Kumar Patra, Harsh Sanklecha},
    pdfborder={0 0 0}
}

\usepackage{float}
\usepackage{tcolorbox}
\usepackage{comment}
\usepackage{relsize}
\usepackage{authblk}
\usepackage{pdfpages}

\newtheorem{theorem}{Theorem}[section]
\newtheorem*{theorem*}{Theorem}
\newtheorem{lemma}[theorem]{Lemma}

\newtheorem{corollary}[theorem]{Corollary}

\newtheorem{observation}[theorem]{Observation}
\newtheorem{definition}{Definition}[section]

\makeatletter
\newcommand{\problemtitle}[1]{\gdef\@problemtitle{#1}}
\newcommand{\probleminput}[1]{\gdef\@probleminput{#1}}
\newcommand{\problemquestion}[1]{\gdef\@problemquestion{#1}}

\NewEnviron{problem}{%
  \problemtitle{}\probleminput{}\problemquestion{}%
  \BODY
  \par\addvspace{.5\baselineskip}
  \noindent
  \begin{tcolorbox}
  \begin{tabularx}{\textwidth}{@{\hspace{\parindent}} l X c}
    \multicolumn{2}{@{\hspace{\parindent}}l}{\@problemtitle} \\
    \textbf{Input:} & \@probleminput \\
    \textbf{Question:} & \@problemquestion
  \end{tabularx}
  \end{tcolorbox}
  \par\addvspace{.5\baselineskip}
}
\makeatother

\usepackage{algorithm}
\usepackage{algpseudocode}
\usepackage{tcolorbox} 

\title{Inclusive and Exclusive Vertex Splitting into Specific Graph Classes: NP Hardness and Algorithms}

\author{
Ajinkya Gaikwad$^{*}$,
Hitendra Kumar$^{*}$,
S. Padmapriya$^{*}$,
Praneet Kumar Patra$^{*}$,
Harsh Sanklecha$^{*}$,
and Soumen Maity%
\thanks{Department of Mathematics, Indian Institute of Science Education and Research (IISER) Pune, India.\\
Emails:\texttt{\{ajinkya.gaikwad, sahu.hitendrakumar, s.padmapriya, praneet.kumar\_patra, harsh.sanklecha\} @students.iiserpune.ac.in}, 
\texttt{soumen@iiserpune.ac.in}}
}

\date{}

\begin{document}

\maketitle

\begin{abstract}
We consider a family of graph modification problems called the $\mathcal{F}$-{\sc Vertex Splitting}, where we are interested in knowing whether, given any graph $G$, it can be modified to a graph $G' \in \mathcal{F}$ through a sequence of at most $k$ vertex splits. Particularly, we consider the $\mathcal{F}$-{\sc Vertex Splitting} when the graph class $\mathcal{F}$ is a constellation, cycle graph, linear forest, and bipartite graphs. We consider two variants of vertex splitting proposed by Abu-Khzam et al. [ISCO 2018], namely \emph{inclusive vertex splitting} and \emph{exclusive vertex splitting}.  We show that the problem $\mathcal{F}$-{\sc Vertex Splitting} is polynomial-time solvable when $\mathcal{F}$ is a cycle graph or a linear forest, for both variants of vertex splitting. When $\mathcal{F}$ is a constellation or a bipartite graph, the $\mathcal{F}$-{\sc Vertex Splitting} is NP-complete for both variants of vertex splitting. 
\end{abstract}


\section{Introduction}~\label{Sec: Introduction}
Graph modification problems have long been a central theme in both classical and parameterized complexity theory. These problems typically aim to transform a given graph into one satisfying certain desired structural properties by applying a limited number of allowed operations. Over the years, a wide range of graph operations-such as vertex deletion, edge contraction, and edge editing have been extensively studied in this context. Each of these operations provides unique insights into the structural behavior of graphs and serves as a foundation for developing efficient algorithms and complexity classifications.

In recent years, a relatively new and intriguing graph operation, known as vertex splitting, has attracted growing attention in the algorithmic and complexity theoretic community. Vertex splitting is a graph modification operation in which a vertex $v$ is removed and replaced by two new vertices $v_1$ and $v_2$, and each neighbor of $v$ is then made adjacent to either $v_1$, $v_2$, or both. Depending on how the adjacency is distributed between the copies, we distinguish between two fundamental variants of this operation: \textit{Inclusive Vertex Splitting} (The neighborhoods of \( v_1 \) and \( v_2 \) may overlap) and \textit{Exclusive Vertex Splitting} (The neighborhoods of \( v_1 \) and \( v_2 \) form a partition of \( N(v) \)).

For each fixed graph class $\mathcal{F}$, we define a distinct vertex splitting problem, denoted as \textsc{$\mathcal{F}$-Inclusive} (respectively \textsc{Exclusive}) \textsc{Vertex Splitting}, \text{ which asks whether a given graph } $G$ can be transformed into a graph satisfying $\mathcal{F}$ by performing at most k inclusive (respectively exclusive) vertex splits. While the inclusive variant allows greater flexibility and can simulate the exclusive one, the exclusive variant is often more restrictive and structurally cleaner. Understanding how these variants differ in their expressive and algorithmic power remains a central challenge in the theory of vertex splitting.

Vertex splitting was initially studied in the context of graph drawings and planarization~\cite{NollenburgEtal2025, FariaFigueiredoMendonca2001, Liebers2024}, where it was used as a tool to reduce edge crossings in dense graphs. Recent works have begun exploring the computational complexity of vertex-splitting based problems under different graph classes and operation variants ~\cite{Abu-KhzamThoumi2025, AhmedKobourovKryven2022, BaumannPfretzschnerRutter2024, Splitthesis, FirbasSorge2024, LiuEtal2025, NollenburgEtal2025}, with the problem of vertex splitting having implications that extend beyond graph drawings and planarization, influencing areas such as circuit design~\cite{PaikReddySahni1998} and statistical analysis~\cite{DavoodiJavadiOmoomi2016, GrammEtal2007}.  Various graph classes, including Cluster~\cite{Splitthesis}, Bicluster~\cite{Abu-KhzamThoumi2025, BentertDrangeHaugen2025}, Claw-Free~\cite{Abu-KhzamThoumi2025}, Bipartite~\cite{FirbasSorge2024}, and Pathwidth-one~\cite{BaumannPfretzschnerRutter2024}, have been extensively studied. However, simpler graph classes, such as constellations, cycle graphs, and linear forests, have been largely overlooked. It is noteworthy that these simpler classes have been analyzed in the context of other graph operations~\cite{GareyJohnsonStockmeyer1974, CardosoKaminskiLozin2007, KochPardalFernandes2024}. In this work, we contribute to this line of research by investigating the complexity of vertex splitting with respect to specific families of graphs, including constellations, cycle graphs, linear forests, and bipartite graphs. 

\textbf{Classical and Parameterized Complexity of the \textsc{$\mathcal{F}$-Inclusive}/\textsc{Exclusive} \textsc{Vertex Splitting} problems in the literature:}
\textsc{Planar Graph-Inclusive Vertex Splitting} is known to be NP-complete \cite{FariaFigueiredoMendonca2001}. Nöllenburg et al. \cite{NollenburgEtal2025} presented an FPT algorithm for solving \textsc{Planar Graph-Inclusive Vertex Splitting}. Baumann et al. \cite{BaumannPfretzschnerRutter2024} showed that \textsc{$\mathcal{F}$-Inclusive Vertex Splitting} is FPT parameterized by the number of splits when either restricted to input graphs of bounded treewidth and properties characterizable by monadic second-order
graph logic, or alternatively, for a minor-closed class $\mathcal{F}$. It is also known that \textsc{Pathwidth-One Inclusive Vertex Splitting} is Fixed Parameter Tractable (FPT) \cite{BaumannPfretzschnerRutter2024}. \textsc{$s$-Club Cluster-Inclusive Vertex Splitting} \cite{Abu-KhzamDavotIsenmannThoumi2026} is NP-complete and is shown to be FPT. Very recently, Abu-Khzam and Thoumi \cite{Abu-KhzamThoumi2025} showed that {\sc Claw-Free Exclusive Vertex Splitting} is NP-complete in general, while admitting a polynomial-time algorithm when the input graph has maximum degree four.

There are some interesting results for \textsc{$\mathcal{F}$-Inclusive Vertex Splitting} for some graph classes $\mathcal{F}$ in the thesis \cite{Splitthesis} of Firbas. \textsc{$\mathcal{F}$-Inclusive Vertex Splitting} is NP-complete \cite{FirbasSorge2024, Splitthesis} for various hereditary graph classes, including, but not limited to, cluster graphs, bipartite graphs, perfect graphs, graphs free of finite sets of (induced) cycles, graphs free of a single biconnected forbidden (induced) subgraph, graphs free of a set of triconnected forbidden (induced) subgraphs with bounded diameter, and graphs free of a set of 4-connected forbidden (induced) subgraphs. \textsc{Forest-, Split-, and Threshold- Inclusive Vertex Splitting} can be solved in polynomial-time \cite{Splitthesis}. \textsc{Cluster-Inclusive Vertex Splitting} \cite{FirbasEtal2025, Splitthesis} is NP-hard and admits a linear kernel.

\subsection{Vertex Splitting Operation and Problem Definitions}~\label{Sec: Problemdefinition}
A \emph{vertex split} is a graph modification in which a vertex  $v$ is replaced by 
two new vertices, $v_1$ and $v_2$. Each edge originally incident to $v$ is reassigned to either $v_1$ or $v_2$. Abu-Khzam et al. \cite{Abu-KhzamEtal2018} proposed two types of vertex splitting operations: Inclusive Vertex Splitting and Exclusive Vertex Splitting. We now formally define both notions below.   

\begin{definition}[Inclusive Vertex Splitting]
Let \( G = (V, E) \) be a graph and \( v \in V \) be a vertex. An \textit{inclusive vertex split} replaces \( v \) with two new vertices \( v_1 \) and \( v_2 \), such that:
\begin{itemize}
    \item \( N(v_1) \cup N(v_2) = N(v) \), where \( N(v) \) is the neighborhood of \( v \);
    \item \( N(v_1) \cap N(v_2) \) may be non-empty, allowing \( v_1 \) and \( v_2 \) to share neighbors.
\end{itemize}
\end{definition}

\begin{definition}[Exclusive Vertex Splitting]
Let \( G = (V, E) \) be a graph and \( v \in V \) be a vertex. An \textit{exclusive vertex split} replaces \( v \) with two new vertices \( v_1 \) and \( v_2 \), such that:
\begin{itemize}
    \item \( N(v_1) \cup N(v_2) = N(v) \);
    \item \( N(v_1) \cap N(v_2) = \emptyset \), ensuring that \( v_1 \) and \( v_2 \) have disjoint neighborhoods.
\end{itemize}
\end{definition}

The key difference between these two operations lies in the treatment of vertex neighborhoods. In Inclusive Vertex Splitting, $v_1$ and $v_2$ are allowed to share neighborhoods, whereas in Exclusive Vertex Splitting, their neighborhoods must be disjoint.

If $v$ is  split into $v_1$ and $v_2$, then $v_1$ and $v_2$ are called the \emph{descendants} of $v$. Conversely, $v$ is called the \emph{ancestor} of $v_1$ and $v_2$. Every vertex split increases the number of vertices in a graph by exactly one. Typically, we are interested in carrying out more than one vertex split. To facilitate this, we introduce the concept of a splitting sequence.

\begin{definition}  
A \textit{splitting sequence} of $k$ splits is a sequence of graphs $G_0, G_1, \ldots, G_k$, such that $G_{i+1}$ is obtainable from $G_i$ via a vertex split for $i \in \{0, \dots, k-1\}$. 
\end{definition}

\noindent The notion of \textit{descendant vertices} (resp. \textit{ancestor vertices}) is extended in a transitive and reflexive way to splitting sequences.

For each fixed graph class $\mathcal{F}$, we define distinct vertex splitting problems as follows.

\begin{center}
\fbox{\begin{minipage}{38.7em}\label{FIVS}
   \textsc{$\mathcal{F}$-Inclusive Vertex Splitting ($\mathcal{F}$-IVS)}\\
   \textbf{Input:} An undirected graph \( G=(V,E) \) and a positive integer \( k \).\\
   \textbf{Question:} Is there a sequence of at most $k$ inclusive vertex splits that transforms $G$ into a graph $G'$ such that $G'\in \mathcal{F}$?
\end{minipage} }
\end{center}

\begin{center}
\fbox{\begin{minipage}{38.7em}\label{FEVS}
   \textsc{$\mathcal{F}$-Exclusive Vertex Splitting ($\mathcal{F}$-EVS)}\\
   \textbf{Input:} An undirected graph \( G=(V,E) \) and a positive integer \( k \).\\
   \textbf{Question:} Is there a sequence of at most $k$ exclusive vertex splits that transforms $G$ into a graph $G'$ such that $G'\in \mathcal{F}$?
\end{minipage} }
\end{center}

In this work, we investigate the computational complexity of {\sc $\mathcal{F}$-Inclusive Vertex Splitting} and {\sc $\mathcal{F}$-Exclusive Vertex Splitting}, where $\mathcal{F}$ corresponds to the class of constellations, cycle graphs, linear forests, and bipartite graphs. When we consider inclusive vertex splitting and $\mathcal{F}$ corresponds to the class of constellations, cycle graphs, linear forests, and bipartite graphs, the resulting problems are referred to as \textsc{Constellation-Inclusive Vertex Splitting (Constellation-IVS)}, \textsc{Cycle Graph-Inclusive Vertex Splitting (Cycle Graph-IVS)}, \textsc{Linear Forest-Inclusive Vertex Splitting (Linear Forest-IVS)}, and \textsc{Bipartite-Inclusive Vertex Splitting (Bipartite-IVS)}, respectively. An analogous naming convention is adopted for the exclusive variant, with the term "Inclusive" replaced by "Exclusive".  

\subsection{Our Results}\label{Sec: Results}

We investigate the classical complexity of {\sc $\mathcal{F}$-Inclusive/Exclusive Vertex Splitting} for different classes $\mathcal{F}$. We observe that {\sc $\mathcal{F}$-Decomposition} and {\sc $\mathcal{F}$-Exclusive Vertex Splitting} are closely related. By using this observation, we prove that {\sc Constellation-Inclusive Vertex Splitting} and {\sc Constellation-Exclusive Vertex Splitting} are NP-complete on cubic graphs and even in planar graphs of degree at most $3$.

We show that both {\sc Cycle Graph-Inclusive Vertex Splitting} and {\sc Cycle Graph-Exclusive Vertex Splitting} are polynomial-time solvable, by showing the relation of these problems with finding closed walks on the graph.

We also show that both {\sc Linear Forest-Inclusive Vertex Splitting} and {\sc Linear Forest-Exclusive Vertex Splitting} are polynomial-time solvable, drawing inspiration from the cycle graph case, but this time searching for open walks rather than closed ones.

Firbas and Sorge \cite{Splitthesis,FirbasSorge2024} showed that {\sc Bipartite-Inclusive Vertex Splitting} is NP-hard by giving a reduction from {\sc 2-Subdivided Cubic Vertex Cover}. We prove that {\sc Bipartite-Inclusive Vertex Splitting} and {\sc Bipartite Vertex Deletion} are equivalent problems. This provides an alternative approach to proving the NP-hardness of {\sc Bipartite-Inclusive Vertex Splitting}. So, the existence of a kernel and FPT algorithm for {\sc Bipartite Vertex Deletion} provides the existence of a kernel and FPT algorithm for {\sc Bipartite-Inclusive Vertex Splitting}. We conclude the same for {\sc Bipartite-Exclusive Vertex Splitting} as we show that  {\sc Bipartite-Exclusive Vertex Splitting} and {\sc Bipartite Vertex Deletion} are equivalent problems.

The Table \ref{ourresults} summarizes the results given in this paper.

\begin{table}[htbp]
\caption{Classical Complexity of various {\sc $\mathcal{F}$-Inclusive/Exclusive Vertex Splitting } problems}\label{ourresults}
\begin{tabular*}{\textwidth}{@{\extracolsep\fill}lcccccc}
\toprule%
Graph Class $\mathcal{F}$ &  $\mathcal{F}$-{\sc Inclusive Vertex Splitting} & $\mathcal{F}$-{\sc Exclusive Vertex Spliitng} \\
\midrule
Constellation & NP-hard (Theorem (\ref{Res:RCIVS})) & NP-hard (Corollary (\ref{Res:RFEVS})) \\
Cycle Graph & Polynomial-time (Theorem (\ref{Res:CGIVS})) & Polynomial-time (Theorem (\ref{Res:CGEVS})) \\
Linear Forest &   Polynomial-time (Corollary (\ref{Res:RLFIVS})) & Polynomial-time (Theorem (\ref{Res:RLFEVS})) \\
Bipartite & NP-hard (Corollary (\ref{Res:BIVS})) & NP-hard (Corollary (\ref{Res:BEVS}))  
\end{tabular*}
\end{table}

\subsection{Preliminaries}\label{Sec: Preliminaries}

Throughout this paper, $G=(V, E)$ denotes a finite, simple, and undirected graph, where $V$ and $E\subseteq V \times V$ represent the vertex and edge set, respectively. An edge $(u,v)\in E$ is also written as $uv$ for convenience. For  $u\in V$, we define \textit{open neighborhood} of $v$ as $N(u)=\{v\in V : (u,v) \in E\}$ and the \textit{closed neighborhood} of $v$ as $N[u]=N(u) \cup \{u\}$.
The  \emph{degree} of $u \in V$ is $|N(u)|$ and denoted by $d_G(u)$. The subgraph induced by $U \subseteq V$ is denoted by $G[U]$. The subgraph induced by $F \subseteq E$ is denoted by $G[F]$.
We define $G-U= G[V\setminus U]$, $G-F=(V,E\setminus F)$ and $G+F=(V,E\cup F)$.
A \textit{path} is a sequence of vertices $v_1, \dots v_m$ such that $\{v_i,v_{i+1}\} \in E(G)$ for $i \in [m-1]$. A graph is called \textit{connected} if there is a path between every pair of vertices. If the graph is disconnected, each of its connected induced subgraphs is called a \textit{component}. A graph is called a \textit{linear forest} if all its components are paths. A \textit{cycle} is a connected graph in which all vertices have degree exactly $2$. A \textit{cubic graph} is a graph that has all vertices of degree exactly $3$. A \textit{tree} is a connected acyclic graph. A \textit{forest} is a graph that has all its components as trees. A \textit{star} is a graph with a single central vertex adjacent to all other vertices. A \textit{bipartite graph} is a graph whose vertex set can be partitioned into two sets such that no two vertices in the same set have an edge between them.\\

A problem is called \emph{fixed-parameter tractable (FPT)} if it can be solved in time of the form $f(k) \cdot \mathrm{poly}(n)$, where $f$ is a computable function that depends only on the parameter $k$, and $n$ denotes the size of the input. A related concept is that of a \emph{kernel}: a problem is said to admit a kernel if there exists a polynomial-time preprocessing algorithm that transforms any instance $(I, k)$ into an equivalent instance $(I', k')$ whose size is bounded solely by some function of $k$. For an in depth reading of these topic, we refer the reader to the book-\cite{bluebook}.

\subsection{\texorpdfstring{$\Gamma$-Decomposition Framework and its connection to Vertex Splitting}{Gamma-Decomposition Framework and its connection to Vertex Splitting}}\label{chap:gammadecomposition}

We now develop the theory of \textit{$\Gamma$-decomposition}, which will be helpful in explaining the reductions in this section.

Let $\Gamma$ be a family of graphs. A $\Gamma$-decomposition of a graph $G$ is a set $\{H_1,H_2,\ldots,H_l\}$ of
subgraphs of $G$ such that the sets $E(H_1),\ldots, E(H_l)$ form a partition of $E(G)$, and each $H_i$ is isomorphic to a graph of $\Gamma$. If \(\Gamma\) contains only one graph, say \(H\), we refer to a  $\Gamma$-decomposition  as an \(H\)-decomposition for simplicity.

Since $\Gamma$-decomposition partitions the edge set of a graph, it is commonly 
assumed in most contexts that $\Gamma$ contains only connected graphs. Throughout this work,  whenever we refer to a $\Gamma$-decomposition,  we assume that $\Gamma$ is a collection of connected graphs. Let $\{H_1,H_2,\ldots,H_l\}$ be a $\Gamma$-decomposition of $G$, where  each $H_i$ is connected. In such cases, knowing the edge set of $H_i$ 
suffices to determine its vertex set, since the vertices are simply the endpoints of the edges in $H_i$. Therefore, it is convenient to describe each $H_i$ by its edge set alone.

Let $E_i=E(H_i)$ for each $1 \leq i \leq l$. Then, for notational convenience, we write $\{E_1,E_2,\ldots,E_l\}$ instead of $\{H_1,H_2,\ldots,H_l\}$.

Formally, we assume that all graphs under consideration have no isolated vertices. This assumption implies that every vertex of $G$ must be included in at least one subgraph of the decomposition.
 
\begin{definition}
For a $\Gamma$-Decomposition $\mathcal{P}= \{H_1,H_2,\ldots,H_l\}$ of $G$, the weight of $\mathcal{P}$ is defined as
$$
{\normalfont\texttt{wgt}}_G(\mathcal{P}) = \sum_{v \in V} \#\mathcal{P}(v),
$$
where $\#\mathcal{P}(v) := |\{i \mid v \in V(H_i), H_i \in \mathcal{P}\}|$ counts the number of sets in $\mathcal{P}$ that contain the vertex $v$.
\end{definition}

If $G$ is a graph, then one can observe that

$$
\texttt{wgt}_G(\mathcal{P}) = \sum_{v \in V} \#\mathcal{P}(v)= \sum_{1\leq i\leq l} |V(H_i)|.
$$

\begin{lemma}\label{Desplitting}
    
Let $G=(V, E)$ be a graph without isolated vertices, and let $\mathcal{P}$ be a $\Gamma$-decomposition of $G$ such that ${\normalfont\texttt{wgt}}_G(\mathcal{P}) \leq |V(G)|+\alpha$. Then, either $G$ is already a disjoint union of graphs of $\Gamma$  or there exists a vertex $u \in V$ such that $u$ can be exclusively split in $G$ to obtain a graph $G'$ satisfying: 
\begin{enumerate}
    \item  $G'$ admits  a $\Gamma$-decomposition, say $\mathcal{P}^\prime$,
    \item  ${\normalfont\texttt{wgt}}_{G'}(\mathcal{P}^\prime) \leq |V(G')|+ \alpha -1$
    \item  $G'$ has no isolated vertices
\end{enumerate}

\end{lemma} 

\proof Suppose $\mathcal{P}= \{H_1,H_2,\ldots,H_l\}$ is a $\Gamma$-decomposition of $G$. If $G$ is not already a disjoint union of  graphs  of $\Gamma$, there must exist $H_i$ and $H_j$ such that $V(H_i) \cap V(H_j) \neq \emptyset$. Without loss of generality, we assume  $H_i=H_1$ and $H_j=H_2$. In this case, let $u \in V(H_1) \cap V(H_2)$.

We define $G'$ as the graph  obtained  by splitting  $u$ into  two vertices $u_{\text {in}}$ and $u_{\text {out}}$ as follows:
$$
\begin{gathered}
N_{G'}(u_{\text {in }}):=N_G(u) \cap V(H_1), \\
N_{G'}(u_{\text {out }}):=N_G(u) \setminus V(H_1).
\end{gathered}
$$

\begin{figure}[H]
    \centering
    \begin{tikzpicture}
    \node[circle,draw,fill=, inner sep=0 pt, minimum size=0.15cm]	(p1) at (-14,14) [label=left:$x$]{};
    \node[circle,draw,fill=, inner sep=0 pt, minimum size=0.15cm]	(p2) at (-14,12) [label=left:$y$]{};
    \node[circle,draw,fill=, inner sep=0 pt, minimum size=0.15cm]	(p3) at (-12,15) [label=right:$a$]{};
    \node[circle,draw,fill=, inner sep=0 pt, minimum size=0.15cm]	(p4) at (-12,13.5) [label=right:$b$]{};
    \node[circle,draw,fill=, inner sep=0 pt, minimum size=0.15cm]	(p5) at (-12,12) [label=right:$u$]{};
    \node[circle,draw,fill=, inner sep=0 pt, minimum size=0.15cm]	(p7) at (-10.5,12.5) [label=right:$d$]{};

    \node[circle,draw,fill=, inner sep=0 pt, minimum size=0.15cm]	(p8) at (-13,11) [label=left:$e$]{};
    \node[circle,draw,fill=, inner sep=0 pt, minimum size=0.15cm]	(p9) at (-11,11) [label=right:$f$]{};

    \node[circle,draw,fill=, inner sep=0 pt, minimum size=0.15cm]	(x1) at (-8,14) [label=left:$x$]{};
    \node[circle,draw,fill=, inner sep=0 pt, minimum size=0.15cm]	(x2) at (-8,12) [label=left:$y$]{};
    \node[circle,draw,fill=, inner sep=0 pt, minimum size=0.15cm]	(x3) at (-6,15) [label=right:$a$]{};
    \node[circle,draw,fill=, inner sep=0 pt, minimum size=0.15cm]	(x4) at (-6,13.5) [label=right:$b$]{};
    \node[circle,draw,fill=, inner sep=0 pt, minimum size=0.15cm]	(x7) at (-4.5,12.5) [label=right:$d$]{};
     \node[circle,draw,fill=, inner sep=0 pt, minimum size=0.15cm]	(x5) at (-6,12) [label=below:$u_{\text{out}}$]{};
     \node[circle,draw,fill=, inner sep=0 pt, minimum size=0.15cm]	(x6) at (-6,11) [label=right:$u_{\text{in}}$]{};

    \node[circle,draw,fill=, inner sep=0 pt, minimum size=0.15cm]	(x8) at (-7,10) [label=left:$e$]{};
    \node[circle,draw,fill=, inner sep=0 pt, minimum size=0.15cm]	(x9) at (-5,10) [label=right:$f$]{};
    
    \draw[orange](p1)--(p3);
    \draw[orange](p4)--(p3);
    \draw[orange](p1)--(p4);
    \draw[red](p1)--(p2);
    \draw[red](p1)--(p5);
    \draw[red](p2)--(p5);
    \draw[green](p7)--(p4);
    \draw[green](p7)--(p5);
    \draw[green](p4)--(p5);
    \draw[blue](p8)--(p5);
    \draw[blue](p9)--(p5);
    \draw[blue](p8)--(p9);

     \draw[orange](x1)--(x3);
    \draw[orange](x4)--(x3);
    \draw[orange](x1)--(x4);
    \draw[red](x1)--(x2);
    \draw[red](x1)--(x5);
    \draw[red](x2)--(x5);
    \draw[green](x7)--(x4);
    \draw[green](x7)--(x5);
    \draw[green](x4)--(x5);
    \draw[blue](x8)--(x6);
    \draw[blue](x9)--(x6);
    \draw[blue](x8)--(x9);

    \node(x0) at (-10.3,11.5) [label=left:$H_1$]{};
    \node(x0) at (-4.3,10.5) [label=left:$H'_1$]{};
    \node(x0) at (-11.7,9) [label=left:$G$]{};
     \node(x0) at (-5.7,9) [label=left:$G'$]{};
    \end{tikzpicture}
    \caption{On the left: a graph $G$ with a $\Gamma$-decomposition where $\Gamma= \{K_3\}$. The subgraph $H_1$ is shown in blue. On the right: the graph $G'$, obtained by splitting $u$ into $u_{\text{in}}$ and $u_{\text{out}}$, along with  a corresponding $\Gamma$-decomposition of $G'$.}
    \label{splitting}
\end{figure}
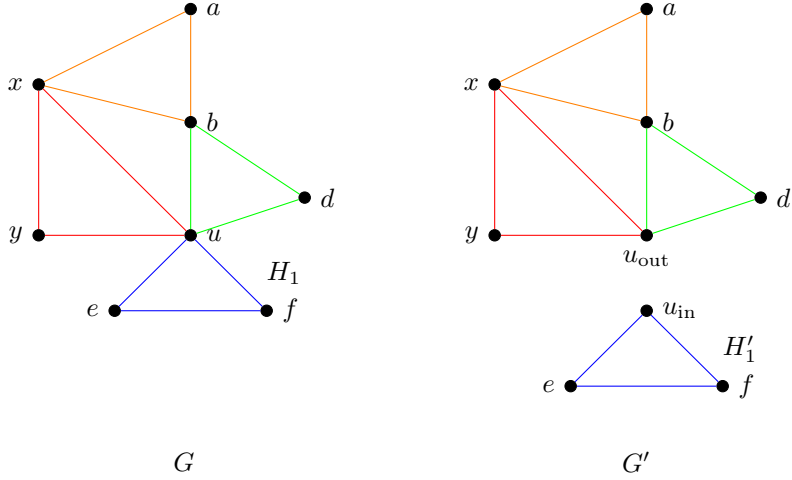

We construct  a $\Gamma$-decomposition of $G'$, denoted by $\mathcal{P}'=
\{H_1^{\prime}, H_2^{\prime}, \ldots,H_l^{\prime}\}$, where:
$$
H'_i:= \begin{cases} H_\text {in} & \text { if } H_i=H_1 
\\ H_{\text {out},i} & \text { if } u \in V(H_i) \wedge (H_i \neq H_1) 
\\ H_i & \text { otherwise }\end{cases}
$$

Here, we define $$H_{\text{in}}= \left((V(H_1) \setminus \{u\}) \cup \{u_\text{in}\},  \{u_{\text{in}}v \mid uv \in E(H_1)\} \cup \{xy \mid xy \in E(H_1), (x \neq u) \land (y \neq u)\}\right)$$
$$H_{\text{out},i}= \left( (V(H_i) \setminus \{u\}) \cup \{u_\text{out}\}, \{u_{\text{out}}v \mid uv \in H_i\} \cup \{xy \mid xy \in E(H_i), (x \neq u) \land (y \neq u)\}\right)$$ 

Informally, $H_{in}$ is the same as $H_1$, except the vertex $u$ is relabeled as $u_{in}$. 

We now verify that $\mathcal{P}'= \{H'_1,H'_2,\ldots,H'_l\}$ is indeed a $\Gamma$-decomposition of $G'$. We first verify that each $H'_i$ is isomorphic to $H_i$, and hence belongs to $\Gamma$.

For $H_1$,  define the map $\phi_1: H_1 \xrightarrow{} H'_1$ by: 

$$
\phi(v)= \begin{cases} u_\text {in} & \text { if } v=u 
\\ v & \text { otherwise }\end{cases}
$$ 
Clearly, $\phi_1$ is an isomorphism from $H_1$ to $H'_1$, hence $H'_1 \cong H_1$.

For any $i \neq 1$ such that $u \in V(H_i)$, define $\phi_{\text{out,i}}: H_i \xrightarrow{} H'_i$  by: 
$$
\phi_{\text{out,i}}= \begin{cases} u_\text {out} & \text { if } v=u 
\\ v & \text { otherwise }\end{cases}
$$ 
This defines an isomorphism from $H_i$ to $H'_i$, and thus $H'_i \cong H_i$.

For all $i$ such that $u \notin V(H_i)$, we have $H_i=H'_i$, so trivially $H'_i \cong H_i$.

Now we have to verify that all the edges of $G'$ are covered by $\mathcal{P}^\prime$. 
Note that the splitting operation does not change the total number of edges. Since each $H'_i \cong H_i$, we have $|E(H_i)|=|E(H'_i)|$. The construction ensures that the $H'_i$ are edge-disjoint. Therefore, the total number of edges covered by $\mathcal{P}^\prime$ is: 

$$\sum_{i}{}|E(H'_i)| = \sum_{i}{}|E(H_i)|= |E(G)|=|E(G')|.$$ 

Thus, $\mathcal{P}^\prime$ is a valid $\Gamma$-decomposition of $G'$.

We now examine the weight of the decomposition $\mathcal{P}^{\prime}$ and prove Condition $2$. We have
\begin{equation*}
\begin{split}
\texttt{wgt}(\mathcal{P}^\prime) &=\sum_{v \in V(G')}\#\mathcal{P}^\prime(v)\\
&= \sum_{v \in V(G') \setminus \{u_{\text{in}},u_{\text{out}}\}}\#\mathcal{P}^\prime (v)+ \# \mathcal{P}^\prime (u_{\text{in}})+ \# \mathcal{P}^\prime(u_{\text{out}}) \\
&= \sum_{v \in V(G) \setminus \{u_{\text{in}},u_{\text{out}}\}}\#\mathcal{P}(v)+ 1 + (\# \mathcal{P}(u)-1)\\
&= \sum_{v \in V(G)}\#\mathcal{P}(v)\\
&\leq |V(G)|+ \alpha\\ 
&= |V(G')|+ \alpha-1 \hfill
\end{split}         
\end{equation*}

Finally, since $G'$ is obtained from $G$ by splitting a vertex, it follows that if $G$ has no isolated vertices, then $G'$ does not either. 
\qed

\section{\textsc{Constellation-Inclusive/Exclusive Vertex Splitting}}~\label{Sec: Star Forest}
A disjoint union of stars is known as a \emph{Constellation}. In this section, we prove that both {\sc Constellation-Inclusive Vertex Splitting} and  {\sc Constellation-Exclusive Vertex Splitting} are NP-hard.  To do so, we establish  the following chain of polynomial-time reductions, starting from the classical NP-hard  \textsc{Vertex Cover} problem: 
$$\textsc{Vertex Cover} \le_P \textsc{Weighted Star Decomposition}$$
$$\le_P \textsc{Constellation-Inclusive Vertex Splitting}$$

We formally define the \textsc{Constellation-Inclusive Vertex Splitting} problem below.\\

\fbox
    {\begin{minipage}{38.7em}\label{CGVS}
       {\sc Constellation-Inclusive Vertex Splitting (Constellation-IVS)}\\
        \noindent{\bf Input:} An undirected graph $G=(V,E)$, and a positive integer $k$.\\
    \noindent{\bf Question:} Is there a sequence of at most $k$ inclusive vertex splitting that transforms $G$ into a constellation?
    \end{minipage} }\\
\vspace{3mm}

If $\Gamma$ is a collection of all star graphs, we will use the term \textit{Star-decomposition} to refer to a $\Gamma$-decomposition. For the sake of completion, we now define the notion of weight for a Star-decomposition.

\begin{definition}
For a Star-decomposition $\mathcal{P}$ of $G$, the weight of  $\mathcal{P}$ is defined as
$$
{\normalfont\text{wgt}}_G(\mathcal{P}) = \sum_{v \in V} \#\mathcal{P}(v),
$$
where $\#\mathcal{P}(v) := |\{P \mid v \in V(P), P \in \mathcal{P}\}|$ denotes the number of sets in $\mathcal{P}$ that contain $v$.
\end{definition}

Now, we formulate the associated decision problem:

\vspace{3mm}
\fbox
    {\begin{minipage}{38.7em}\label{WSD}
       {\sc Weighted Star-Decomposition}\\
        \noindent{\bf Input:} An undirected graph $G=(V,E)$, and a positive integer $k$.\\
    \noindent{\bf Question:} Does there exist  a Star-decomposition $\mathcal{P}$ of $G$ such that $\texttt{wgt}(\mathcal{P}) \leq k$?
    \end{minipage} }\\
    
\begin{theorem}
   {\sc  Weighted Star Decomposition} is NP-complete on cubic graphs and also on planar graphs of maximum degree $3$.
\end{theorem}

\proof It is easy to see that {\sc Weighted Star-Decomposition} is in NP. It is known that  {\sc Vertex Cover} is NP-complete on cubic graphs and also on planar graphs of maximum degree $3$. To prove that {\sc Weighted Star-Decomposition} is NP-hard on cubic graphs and also on planar graphs of maximum degree $3$, we give a reduction from  {\sc Vertex Cover}. We prove that $(G,k)$ is a yes-instance of {\sc Vertex Cover} if and only if $(G,m+k)$ is a yes-instance of {\sc Weighted Star-Decomposition}.

For the forward direction, we assume that $(G,k)$ is a yes-instance of {\sc Vertex Cover}. Suppose $X=\{v_1,v_2, \ldots, v_l\}$ is a vertex cover of $G$ with $l\leq k$. Now, we define $E_1= \{v_1w : w \in N(v_1) \}$ and $E_i= \{v_iw : w \in N(v_i) \setminus \{v_1,v_2,\ldots,v_{i-1} \} \} \text{ for } 2 \leq i \leq l$. We claim that $\mathcal{P}= \{G[E_1],G[E_2],\ldots,G[E_l]\}$ is a star decomposition with $\texttt{wgt}(\mathcal{P}) \leq m+k$.

Clearly, each $E_i$ is a star with center vertex $v_i$. Note that $G-X$ is an independent set, so all the edges of graph $G$ are either between $X$ and $G-X$ or between the vertices of $X$. The way the construction of $E_i$ is done and $X$ is a vertex cover makes it clear that $\mathcal{P}$ covers all edges of $G$. It is also clear that $E_i \cap E_j = \emptyset$ whenever $i \neq j$.

Note that if $E_i$ has $p_i$ edges, then the number of vertices in $E_i$ is $p_i+1$. Hence,
$$\texttt{wgt}(\mathcal{P})= \sum_{1 \leq i \leq l} |V(E_i)|= \sum_{1 \leq i \leq l} (p_i+1)= \sum_{1 \leq i \leq l}p_i+ l=m+l \leq m+k.$$  

For the reverse direction, we assume that there exists a star decomposition, say $\mathcal{S}$ of $G$ with $\texttt{wgt}(\mathcal{S})\leq m+k$. Suppose $\mathcal{S}=\{S_1,S_2, \ldots,S_q\}$.

Define $C=\{c_1,c_2,\ldots,c_q\}$ where $c_i$ is center of star $S_i$. Suppose a star with center $c_i$ contains $p_i$ number of edges, then the number of vertices is $p_i+1$ in that star. Now we see that $\texttt{wgt}(\mathcal{S})= \sum_{1 \leq i \leq q} (p_i+1)= \sum_{1 \leq i \leq q}p_i+q=m+q.$ Since $\texttt{wgt}(\mathcal{S}) \leq m+k$, we have $q \leq k$.

We claim that $C=\{c_1,c_2,\ldots,c_q\}$ is a vertex cover of $G$. Let $uv$ be an edge of the graph $G$. Since $\mathcal{S}$ is a sigma star cover of $G$, $uv$ must be in $S_i$ for some $i$. Since $S_i$ is a star graph, either $u$ or $v$ must be $c_i$. The edge $uv$ is covered by $C$ because $C$ contains $c_i$. So we have proved that $C$ is a vertex cover of size at most $k$, hence $(G,k)$ is a yes-instance of the {\sc Vertex Cover}.
\qed \\

Notice that {\sc Weighted Star Decomposition} is NP-hard for graphs without isolated vertices as well.

Now we aim to prove that {\sc Constellation-Vertex Splitting} is NP-hard by giving a polynomial-time reduction from {\sc Weighted Star Decomposition}. But we first prove a few lemmas.

\begin{lemma}\label{edgemerging}
Let $G=(V, E)$ be a graph, and let $G'=(V', E')$ be obtained from $G$ by inclusive vertex splitting of a vertex $u \in V$ into two vertices $v, w \in V'$. If $\mathcal{P}'$ is a star decomposition of $G'$, then there exists a star decomposition $\mathcal{P}$ of $G$ with ${\normalfont\texttt{wgt}}(\mathcal{P}) \leq  {\normalfont\texttt{wgt}} (\mathcal{P}')$.
\end{lemma}

\proof Let $\mathcal{P}'=\{G'[E'_1],G'[E'_2],\ldots,G'[E'_t]\}$ be a star decomposition of $G'$. Then using 
$$
f(E'_i):= \begin{cases} (E'_i \setminus \{vz,wz \mid z \in V(E'_i)\}) \cup \{uz \mid vz \in E'_i \lor wz \in E'_i  \} & \text {if } V(E'_i) \cap\{v, w\} \neq \emptyset \\ 
E'_i & \text {otherwise }\end{cases}
$$
we define
$$
\mathcal{P}:=\left\{G[E_1],G[E_2], \ldots,G[E_t] \mid f\left(E'_i\right) = E_i \text{ for } 1\leq i \leq t \right\}.
$$
We claim that $\mathcal{P}$ satisfies the conditions of this lemma. First, we establish that $\mathcal{P}$ is a star decomposition of $G$. 

We begin by proving that all $G[E_i] \in \mathcal{P}$ are stars.

Let $G[E_i] \in \mathcal{P}$. Assume $f^{-1}(E_i)=E_i$. Observe that $E_i \cap\{v, w\}=\emptyset$. Clearly, $G[E_i]$ is a star.

Conversely, assume $f^{-1}(E_i) \neq E_i$. Then two cases can occur.\\

\textbf{Case 1}: $G'[f^{-1}(E_i)]$ contains both $v$ and $w$.
By assumption, $G'[f^{-1}(E_i)]$ is a star. Suppose $G'[f^{-1}(E_i)]$ contains both $v$ and $w$. By definition of vertex splitting, it follows that $v$ and $w$ can not be the center of the star $G'[f^{-1}(E_i)]$. Let $r$ be the center of $G'[f^{-1}(E_i)]$. Then we see that $G[E_i]$ is also a star with center $r$ and contains $u$.\\

\textbf{Case 2}: $G'[f^{-1}(E_i)]$ contains exactly one of $v$ and $w$. Without loss of generality, we assume that $G'[f^{-1}(E_i)]$ contains $v$ and not $w$. Then $E_i= (E'_i \setminus \{vz \mid vz \in E'_i\}) \cup \{uz \mid vz \in E'_i \}$. It is clear that $E'_i$ is isomorphic to $E_i$. Hence, $E_i$ is a star. \\

Now, we prove the second property, that is, all edges of $G$ are covered by $\mathcal{P}$. Let $v_1 v_2 \in E(G)$.\\

\textbf{Case 1}: $\{v_1, v_2\} \cap\{u\}=\emptyset$: The edge $v_1v_2$ is not affected by the split, therefore $v_1 v_2 \in E(G')$. Hence there is some $E'_i \in \mathcal{P}'$ such that $v_1v_2 \in E'_i$. Thus $v_1v_2 \in E_i=f(E'_i)$.\\

\textbf{Case 2}: $\{v_1, v_2\} \cap\{u\} \neq \emptyset$ : Without loss of generality, assume $v_1=u$. The definition of vertex splitting tells that either $v v_2 \in E(G')$ or $w v_2 \in E(G')$ must hold. Without loss of generality, assume the former. Since $\mathcal{P}'$ is a star-decomposition of $G'$ there exist some $G'[E'_j] \in \mathcal{P}'$ such that $v v_2 \in E'_j$. Thus, we find that $u,v_2 \in E_j=f(E'_j)$ and $G[E_j] \in \mathcal{P}$.\\

Therefore, $\mathcal{P}$ is a Star-decomposition of $G$. We observe that $f$ ranging over $\mathcal{P}'$ does not increase the cardinality of any image it maps. Hence we have, $\texttt{wgt}(\mathcal{P}) \leq \texttt{wgt}\left(\mathcal{P}'\right)$.
\qed \\ 

Since exclusive vertex splitting is also inclusive vertex splitting, the above proof works when we do exclusive vertex splitting of $u$.

\begin{theorem}\label{equivalency}
   Let $G=(V, E)$ be a graph without any isolated vertices. Then $(G,k)$ is a yes-instance of {\sc Constellation-Inclusive Vertex Splitting} if and only if $(G,|V(G)|+k)$ is a yes-instance of {\sc Weighted Star Decomposition}.
\end{theorem}

\proof Let $(G,k)$ be a yes-instance of {\sc Constellation-Inclusive Vertex Splitting}, that is, there exists a sequence of at most $k$ inclusive vertex splitting in $G$ such that the resulting graph $G'$ is a disjoint union of stars.

Let $G_0, \ldots, G_{l}$ be a sequence of graphs with $G_0=G$ and $l \leq k$ such that each graph, except $G_0$, is obtained from its predecessor via an inclusive vertex split, and $G_{l}$ is a disjoint union of stars. By identifying all connected components of $G_{l}$ with their vertex sets, we can construct a star decomposition $\mathcal{P}_{l}$ of $G_{l}$ with \texttt{wgt}$(\mathcal{P}_l)=|V(G_l)|$. Each split used in the construction of $G_1, \ldots, G_{l}$ introduces exactly one new vertex, therefore $|V(G_l)|=|V(G)|+l$. Combining this with the fact that $l \leq k$, we derive \texttt{wgt}$(\mathcal{P}_l) \leq |V|+k$. Using the sequence $G_0, \ldots, G_{l}$ in reverse order, we iteratively apply Lemma \ref{edgemerging} $l$ times using $\mathcal{P}_{l}$ and $G_{l}$ as base case and obtain $\mathcal{P}_0, \ldots, \mathcal{P}_{l}$. In particular, it follows that $\mathcal{P}_0$ is a star decomposition of $G$ satisfying \texttt{wgt}$(\mathcal{P}_0)\leq |V(G)|+k $. Thus, $(G,|V(G)|+k)$ is a yes-instance.

Let $(G,|V(G)|+k)$ be a yes-instance of {\sc Weighted Star Decomposition}, that is, there exists a star decomposition, say $\mathcal{P}=\{ H_1, H_2, \dots, H_p\}$ with $\texttt{wgt}(\mathcal{P}) \leq n+k$. By iteratively applying Lemma \ref{Desplitting} for a number of times, call it $l$, either until a constellation is obtained as a direct result of the lemma, or alternatively, stopping after $k$ iterations, we can obtain the sequences $G_0, \ldots, G_{k}$ and $\mathcal{P}_0, \ldots, \mathcal{P}_{k}$. Where $G_i \forall 0 \geq i \geq k$ have no isolated vertices. 

We shall now verify that $G_{k}$ must be a disjoint union of stars. As a consequence of the applications of Lemma \ref{Desplitting}, we get $\texttt{wgt}(\mathcal{P}_{k}) \leq |V(G_{k})|$. By considering the fact that for each vertex $v \in V\left(G_{k}\right)$ there exists a subgraph, say $H_i^k \in \mathcal{P}_k$ with $v \in E_i^k$ since $G_k$ has no isolated vertices, we derive $\texttt{wgt}(\mathcal{P}_k)\geq |V(G_{k})|$. Thus, we have that \texttt{wgt}$(\mathcal{P}_{k})=|V(G_{k})|$ and it follows that $G_{k}$ is a disjoint union of stars. Hence $(G,k)$ is a yes-instance of {\sc Constellation-Inclusive Vertex Splitting}.
\qed \\

\noindent \textbf{Remark:} 
For ease of explanation, we considered a graph with no isolated vertices. However, we note that our method also applies to graphs with isolated vertices. The only crucial observation is that, at any stage, a vertex split that separates a vertex $v$ into an isolated vertex $u$ and another vertex $w$ can be safely avoided. This is because $w$ is essentially just a relabeling of $v$. 

\begin{theorem}\label{Res:RCIVS}
    {\sc Constellation-Inclusive Vertex Splitting} is NP-complete on cubic graphs and also on planar graphs of maximum degree $3$.
\end{theorem}

\proof It is easy to see that {\sc Constellation-Inclusive Vertex Splitting} is in NP, because a certificate of {\sc Constellation-Inclusive Vertex Splitting} can be guessed and checked in polynomial time.
By Theorem \ref{equivalency}, we conclude that deciding a $(G,k)$ instance of {\sc Constellation-Inclusive Vertex Splitting} is equivalent to deciding the instance $(G, k - |V(G)|)$ of {\sc Weighted Star Decomposition}. Since {\sc Weighted Star Decomposition} is NP-complete on cubic graphs and even on planar graphs of maximum degree $3$, {\sc Constellation-Inclusive Vertex Splitting} is NP-hard on cubic graphs and also on planar graphs of maximum degree $3$. Consequently, we conclude that {\sc Constellation-Inclusive Vertex Splitting} is NP-complete on cubic graphs and also on planar graphs of maximum degree $3$.\qed 

Note that the proof of Theorem \ref{equivalency} also applies to the case of exclusive vertex splitting. In the forward direction of the proof, we used Lemma \ref{edgemerging}, and as observed, Lemma \ref{edgemerging} also applies to exclusive vertex splitting. Hence, there is no issue in the forward direction of the proof of Theorem \ref{equivalency}. In the backward direction, we used Lemma \ref{Desplitting}, which involves exclusive vertex splitting. Based on this, we now define the following problem,\\

\fbox
    {\begin{minipage}{38.7em}\label{SFEVS}
       {\sc Constellation-Exclusive Vertex Splitting (Star Forest-EVS)}\\
        \noindent{\bf Input:} An undirected graph $G=(V,E)$, and a positive integer $k$.\\
    \noindent{\bf Question:} Is there a sequence of at most $k$ exclusive vertex splitting that transforms $G$ into a constellation?
    \end{minipage} }\\

Then, we have the following two results.

\begin{corollary}
     Let $G=(V, E)$ be a graph. Then $(G,k)$ is a yes-instance of {\sc Constellation-Exclusive Vertex Splitting} if and only if $(G,|V(G)|+k)$ is a yes-instance of {\sc Weighted Star Decomposition}
\end{corollary}

\begin{corollary}\label{Res:RFEVS}
    The {\sc Constellation-Exclusive Vertex Splitting} problem is NP-complete on cubic graphs and also on planar graphs of maximum degree $3$.
\end{corollary}

\section{Cycle Graph-Inclusive/Exclusive Vertex Splitting}

We say a graph is a \emph{cycle graph} if it is a disjoint union of cycles. We now formally define the problems considered in this section.

\begin{center}
\fbox
    {\begin{minipage}{38.7em}\label{CGEVS}
       {\sc Cycle Graph-Exclusive Vertex Splitting (Cycle Graph-EVS)}\\
        \noindent{\bf Input:} An undirected graph $G=(V,E)$, and a positive integer $k$.\\
    \noindent{\bf Question:} Is there a sequence of at most $k$ exclusive vertex splittings that transforms $G$ into a cycle graph?
    \end{minipage} }\\
\end{center}

\begin{center}
\fbox
    {\begin{minipage}{38.7em}\label{CGIVS}
       {\sc Cycle Graph-Inclusive Vertex Splitting (Cycle Graph-IVS)}\\
        \noindent{\bf Input:} An undirected graph $G=(V,E)$, and a positive integer $k$.\\
    \noindent{\bf Question:} Is there a sequence of at most $k$ inclusive vertex splittings that transforms $G$ into a cycle graph?
    \end{minipage} }\\
\end{center}

\begin{definition}
    A cycle decomposition of a graph $G$ is a set of cycles in $G$ such that every edge of $G$ belongs to exactly one cycle.
\end{definition}

We note that if $\Gamma = \{C_3, C_4, C_5,\ldots \}$, where $C_l$ denotes a cycle of length $l$ for all $l\geq 3$, then a $\Gamma$-decomposition and a cycle decomposition are the same notions for a graph.

The following known result will be relevant to our discussion:

\begin{lemma} (\cite{Euler1741,Petersen1891})\label{evendegree}
    A graph $G$ has a cycle decomposition if and only if every vertex of $G$ has even degree.
\end{lemma}

\begin{lemma}\label{alleven}
    If $(G,k)$ is a yes-instance of {\sc Cycle Graph-Exclusive Vertex Splitting}, then every vertex of $G$ has even degree.
\end{lemma}

\proof  Let $G_0,\ldots , G_l$ be a splitting sequence such that $G_0 = G$ and $G_l$ is a cycle graph. Since we perform only exclusive vertex splitting, we have 

$$d_G(v)= \sum_{v': v' \text{ is a descendant of  $v$ in } G_l} d_{G_l}(v').$$

Here, we say $v'$ is a descendant of $v$ if $v'$ is obtained by splitting $v$.

Since $G_l$ is a cycle graph, every vertex in $G_l$ has even degree. Therefore, every vertex in $G$ has even degree.\qed \\

By applying Lemma \ref{alleven}, we obtain the following corollary.

\begin{corollary}\label{corollary}
If $G$ has a vertex with odd degree, then $(G,k)$ is a no-instance of {\sc Cycle Graph-Exclusive Vertex Splitting} for any $k\geq0$.
\end{corollary}

\begin{theorem}\label{maintheorem}
   Let $G$ be a graph where every vertex has even degree. then $(G,k)$ is a yes-instance of {\sc Cycle Graph-Exclusive Vertex Splitting} if and only if $k\geq m-n$, where $|V(G)|=n$ and $|E(G)|=m$.
\end{theorem}

\proof For the forward direction, suppose $(G,k)$ is a yes-instance. For the sake of contradiction,  assume that $k<m-n$. Let $G_0,\dots , G_t$ be a splitting sequence such that $G_0 = G$, $G_t$ is a cycle graph and $t\leq k$. So we have $|V(G_t)|=|V(G)|+t \leq n+k < n+(m-n)=m$. Since $G_t$ is a cycle graph, it satisfies $|E(G_t)|=|V(G_t)|<m$. But since we are doing exclusive vertex splitting, the number of edges remains constant throughout the sequence, so  $|E(G)|=|E(G_t)|=m$, which is a contradiction. Hence $k \geq m-n$.

For the backward direction, let $k\geq m-n$. By assumption, every vertex of $G$ has even degree. By Lemma \ref{evendegree}, we have a cycle decomposition, say $\mathcal{P}$, of $G$. Note that every cycle contributes $2$ to the degree of the vertices it passes through. Hence, every vertex with degree $d(v)$ will be part of exactly $\frac{d(v)}{2}$ elements of $\mathcal{P}$. Therefore, we have $\texttt{wgt}(\mathcal{P})= \sum_{v \in V(G)} \frac{d(v)}{2}=m$. By iteratively applying Lemma \ref{Desplitting} $m-n$ number of times, we can obtain the sequences $G_0, \ldots, G_{m-n}$ and $\mathcal{P}_0, \ldots, \mathcal{P}_{m-n}$.

We claim that $G_{m-n}$ is a disjoint union of cycles (i.e., a cycle graph). As a consequence of the applications of Lemma \ref{Desplitting}, we get $\texttt{wgt}(\mathcal{P}_{m-n}) \leq |V(G_{m-n})|$. By considering the fact that for each vertex $v \in V\left(G_{m-n}\right)$ there exists $H_i \in \mathcal{P}_{m-n}$ with $v \in V(H_i)$, we derive $\texttt{wgt}(\mathcal{P}_{m-n})\geq |V(G_{m-n})|$. Thus, we have that \texttt{wgt}$(\mathcal{P}_{m-n})=|V(G_{m-n})|$ and it follows that $G_{m-n}$ is a disjoint union of cycles. Hence $(G,m-n)$ is a yes-instance of {\sc Cycle Graph-Vertex Splitting}. So it is obvious that $(G,k)$ is a yes-instance. \qed \\

We note that Corollary \ref{corollary} deals with graphs that have at least one vertex with odd degree, and Theorem \ref{maintheorem} deals with those graphs where all vertices have even degrees. Therefore, by combining Corollary \ref{corollary} and Theorem \ref{maintheorem}, we obtain the following final result.

\begin{theorem}\label{Res:CGEVS}
The {\sc Cycle Graph-Exclusive Vertex Splitting} is solvable in polynomial time.
\end{theorem}

We note that there is no way to obtain a cycle graph from an isolated edge by performing a sequence of exclusive vertex splits. However, the situation changes if we allow inclusive vertex splits. To see this, consider an isolated edge $(x,y)$. By first splitting $x$ inclusively, where both $y_1$ and $y_2$  are adjacent to $x_1$ and $x_2$, we obtain a cycle of length $4$.

We show that \textsc{Cycle Graph-Inclusive Vertex Splitting} is also polynomial-time solvable. However, the proof technique differs from the one used for \textsc{Cycle Graph-Exclusive Vertex Splitting}. We begin with a simple yet crucial observation that will play an important role in the proof.

\begin{observation}\label{splitnumber}
\label{min edges}
    Let $G$ be a graph and $G'$ be a cycle graph obtained using a sequence of inclusive vertex splits from $G$. Then the number of splits required to obtain $G'$ from $G$ is exactly $|E(G')|-|V(G)|$. 
\end{observation}

\begin{proof}
    Since $G'$ is a cycle graph, the number of vertices and edges of $G'$ is the same. As each of the splits introduces exactly one new vertex, the number of inclusive vertex splits required to obtain $G'$ from $G$ is exactly $|V(G')|-|V(G)|$, which is the same as $|E(G')|-|V(G)|$.
\end{proof}

\begin{lemma}\label{openingwalk}
Let $G=(V, E)$ be a graph and let $W$ be a closed walk that visits all vertices and covers all the edges in $G$. If $|W|$ denotes the length of the walk $W$, then there exists a sequence of exactly $|W|-|V(G)|$ inclusive vertex splits such that the resulting graph is a cycle graph (a single cycle). 
\end{lemma}

\begin{proof}
   Let $W=(w_0,w_1\cdots w_{l-1},w_l)$ be a closed walk in a graph $G$, where $w_l=w_0$. Note that the walk may traverse the same edge multiple times.
   
   The following provides a procedure for generating a cycle graph using a sequence of inclusive vertex splits. We start our procedure by coloring all the vertices of the graph $G$ black. During the procedure, for a vertex $v \in V(G)$,  $v^b$ denotes the copy of $v$ that is black and $v^w$ denotes the copy of $v$ that is white. 

   We construct the multigraph $G_M=(E_M, V)$ from $G$ where the number of parallel edges between two vertices $u,v \in V(G_M)$ is the number of times the edge $(u,v)$ has appeared in the walk $W$. Note that we could have done the procedure we mentioned now, even without constructing this multigraph. The multigraph essentially helps us to achieve in exclusive splits what we wanted to achieve through inclusive. For the inclusive version, one might replace the step of given parallel edges to two split vertices with the step of making both the split vertices adjacent to the same vertex. Writing through the inclusive vertex lens is easier, so we stuck with it. We store a temporary variable \texttt{prev}, and we initialize this to $w_{l-1}$.
   
   We proceed the walk sequentially from $w_0$ to $w_{l-1}$. When the vertex $w_j$ is visited, we perform the following
   \begin{enumerate}
        \item If $j = l-1$ or the degree of the current vertex \( w_j \) in \( G_M \) is exactly 2, we just change the color of the vertex from black to white. \texttt{prev} is reassigned to this white vertex.
       \item If $j \neq l-1$ and the degree of \( w_j \) in \( G_M \) is at least 4, we exclusively split ${w_j}$ into two vertices one white ${w^w_j}$, another black ${w^b_j}$ where the white copy is only adjacent to \texttt{prev} and ${w^b_{j+1}}$, we reassign \texttt{prev} as the white vertex just created. What this also means is, in our procedure below, if a vertex in the walk $W$, say $w_i$ is the same as $w_j$ for $i<j$, when we reach $w_j$, it would be a black descendant of $w_i$,       
   \end{enumerate}

The above steps are mentioned as algorithm~\ref{Algo: Constructing a cycle graph from a walk}. Call the modified graph returned by this algorithm $G'$. One can check that the number of edges in $ G'$ is $ |W|$, and by construction, the resultant connected graph has all vertices as degree 2 vertices. Thus, the final graph is a collection of cycles generated by a sequence of exclusive vertex splits on $G$, thus ensuring that the number of splits is $|W|-|V(G)|$.

\begin{algorithm}
\caption{Constructing a cycle graph from a walk $W$}
\label{Algo: Constructing a cycle graph from a walk}
\begin{algorithmic}
\State \textbf{Input:} Closed walk $W=(w_0,w_1,\dots,w_{l-1},w_l)$ in graph $G$
\State \textbf{Initialize:} Color all vertices of $G$ black
\State Construct multigraph $G_M=(V,E_M)$ where multiplicity of edge $\{u,v\}$ equals its number of occurrences in $W$
\State $\texttt{prev} \gets w_{l-1}$
\State $j \gets 0$
\While{walk not exhausted}
    \If{$j = l-1$ \textbf{ or } $\deg(w_j) = 2$}
        \State Change color of $w_j$ from black to white
        \State $prev \gets$ $w_j$
    \ElsIf{$j \neq l-1$ \textbf{ and } $\deg(w_j) \geq 4$}
        \State Do an exclusively vertex splitting of $w_{j}$ into a white and black copy, $w_j^b$ and $w_j^w$ respectively, making $w_j^w$ adjacent to $\texttt{prev}$ and $w^b_{(j+1) \bmod l}$, thus making $w_j^b$ adjacent to the remaining neighbors of $w_j$.
        \State $\texttt{prev} \gets$ $w_j^w$
    \EndIf
    \State $j \gets (j+1) \bmod l$
\EndWhile
\State \textbf{Output:} Graph $G'$ with $|W|$ edges, which is a cycle graph.
\end{algorithmic}
\end{algorithm}
\end{proof}

First, we observe that if a certain number of inclusive vertex splits is enough to get a cycle graph from a graph, then the same number of inclusive vertex splits is also enough to obtain a single cycle.

\begin{lemma}\label{equalsplits}
    Given a graph $G = (V, E)$. Then, there exists a sequence of $k$ inclusive vertex splits that transforms $G$ into a disjoint union of cycles if and only if there exists a sequence of $k$ inclusive vertex splits that transforms $G$ into a single cycle. 
\end{lemma}

\begin{proof}
    For the forward direction, let us assume that there exists a sequence of $k$ inclusive vertex splits that transforms $G$ into a disjoint union of cycles, say $G'$.
    
    By Observation \ref{splitnumber}, $|E(G')|-|V(G)|=k$.
    If $G'$ is a single cycle, then we are done.
    So, we assume that $G$ is a disjoint union of at least two cycles, say $C_1,C_2,\ldots,C_r$.
    Each cycle in $G'$ corresponds to a closed walk in the graph $G$. For each $i \in [r]$, let $T_i$ denote the closed walk corresponding to the cycle $C_i$ in $G$.
    These walks may traverse the same edge multiple times, and different walks may share edges or vertices.
    
    Construct a new multigraph $\Tilde{G}=(V,\Tilde{E})$ as follows: for every edge $e \in E(G)$ and for every occurrence of $e$ as an edge in some walk $T_i$, include a distinct parallel copy of $e$ in $\Tilde{E}$.
    Then each $T_i$ becomes a trail in $\Tilde{G}$ and the trails corresponding to different $T_i$ are edge-disjoint in $\Tilde{G}$.
    This implies that every vertex in the graph $\Tilde{G}$ is even. Hence $\Tilde{G}$ is an Eulerian graph.
    
    We find an Eulerian trail in the graph $\Tilde{G}$. By using the trial, we obtain a single closed walk, say $W$, in the original graph, where each parallel edge in $\Tilde{G}$ is replaced by its corresponding original edge in $G$. $W$ has exactly as many edges (counted with multiplicity) as there are edges in $G'$.
    By Lemma \ref{openingwalk}, we can obtain a single cycle, say $C'$, by $|W|-|V(G)|$ number of inclusive vertex splits. We have already observed that $|W|=|E(G')|$.
    Hence, we need exactly $|E(G')|-|V(G)|$, which is equal to $k$, the number of splits to get a single cycle. Backward direction is trivial.\end{proof}

As a consequence of Observation \ref{splitnumber} and Lemma \ref{equalsplits}, we conclude that for a graph $G$, the problem of determining the minimum number of inclusive vertex splits required to obtain a cycle graph reduces to the problem of finding the smallest closed walk in $G$ that traverses every edge at least once. The problem of finding the smallest closed walk that visits every edge is precisely the well-studied {\sc Chinese Postman}, which is polynomial-time solvable ~\cite{EdmondsJohnson1973}. Hence, we get the following theorem.

\begin{theorem}\label{Res:CGIVS}
    The {\sc Cycle Graph-Inclusive Vertex Splitting} is solvable in polynomial time.
\end{theorem}

\section{\textsc{Linear Forest-Inclusive/Exclusive Vertex Splitting}}~\label{Sec: Linear Forest}
The aim of this section is to give a polynomial-time algorithm for {\sc Linear Forest-Inclusive Vertex Splitting} and {\sc Linear Forest-Exclusive Vertex Splitting}. Drawing an analogy with the previous section, we note that the focus of attention now shifts from searching for closed walks to searching for trails. We first define our problems formally.\\

\fbox
    {\begin{minipage}{38.7em}\label{LFEVS}
       {\sc Linear Forest-Exclusive Vertex Splitting (Linear Forest-EVS)}\\
        \noindent{\bf Input:} An undirected graph $G=(V,E)$, and a positive integer $k$.\\
    \noindent{\bf Question:} Is there a sequence of at most $k$ exclusive vertex splittings that transforms $G$ into a linear forest?
    \end{minipage} }
\vspace{3mm}

\fbox
    {\begin{minipage}{38.7em}\label{LFIVS}
       {\sc Linear Forest-Inclusive Vertex Splitting (Linear Forest-IVS)}\\
        \noindent{\bf Input:} An undirected graph $G=(V,E)$, and a positive integer $k$.\\
    \noindent{\bf Question:} Is there a sequence of at most $k$ inclusive vertex splittings that transforms $G$ into a linear forest?
    \end{minipage} }
\vspace{3mm}

\begin{observation}\label{exisin}
    The minimum number of exclusive splits required is the same as the minimum number of inclusive splits required to convert a graph $G$ to a linear forest (all the connected components are paths with at least one edge).
\end{observation}

\begin{proof}
Let $k_{\mathrm{in}}(G)$ (respectively $k_{\mathrm{ex}}(G)$) be the minimum number of inclusive (respectively exclusive) vertex splits required to obtain a linear forest from the graph $G$.

Since every exclusive vertex split is a special case of an inclusive vertex split, we have $k_{\mathrm{in}}(G) \leq k_{\mathrm{ex}}(G)$.
  
We claim that $k_{\mathrm{ex}}(G) \leq k_{\mathrm{in}}(G)$. For the sake of contradiction, we assume that $k_{\mathrm{in}}(G) < k_{\mathrm{ex}}(G)$.

Suppose we construct a linear forest $G_P'$ from $G$ using $k_{in}$ inclusive vertex splits. We can simulate this sequence using exclusive vertex splits while preserving the number of splits. Specifically, whenever a vertex $v$ is split into $v_1$ and $v_2$ in the inclusive version, where both are adjacent to a vertex $u$ (to which $v$ was originally adjacent), we modify the step in the exclusive version by keeping only one of $v_1$ or $v_2$ adjacent to $u$ (chosen arbitrarily).

This might cause a situation where, in the inclusive sequence, a vertex $x$ was split into $x_1$ and $x_2$, with at least one of them adjacent to some vertex $y$, while in the exclusive version, $x$ was not adjacent to $y$ at all. In such cases, we ensure that neither $x_1$ nor $x_2$ is adjacent to $y$ in the exclusive version. This entire description is written as algorithm~\ref{Alg: Simulating Inclusive Vertex Splits Using Exclusive Vertex Splits} where we assume $v_1,v_2,\cdots,v_{k_{in}}$ is the sequence of inclusive vertex split that transforms $G$ to $Gp'$ such that after the first $l$ splits (in the sequence mentioned), the graph generated is ${G_p^{l}}'$ evidently, ${G_p^{k_{in}}}'=G_p'$ and with a similar use of notation we denote the graph in our algorithm after $k$ exclusive as ${G^{k}}''_P $.  

\begin{algorithm}
\caption{Simulating Inclusive Vertex Splits Using Exclusive Vertex Splits}
\label{Alg: Simulating Inclusive Vertex Splits Using Exclusive Vertex Splits}
\begin{algorithmic}
\For{$i \in [k_{\mathrm{in}}]$}
    \State Let $v_i$ be the vertex to be split in the $i$-th inclusive split.
    \If{$v_i \to v_i^1, v_i^2$ with both $v_i^1, v_i^2$ adjacent to some 
         $u \in N_{{G_P^{i-1}}'}(v_i)$ such that $u \in N_{{G_P^{i-1}}''}(v_i)$}
        \State Split $v_i \in V({G_P^{i-1}}'')$ into $v_i^1, v_i^2$.
        \State For each such $u$, make \emph{either} $v_i^1$ or $v_i^2$ adjacent to $u$ in ${G_P^{i-1}}''$.
    \ElsIf{$v_i \to v_i^1, v_i^2$ with at least one of $v_i^1$ or $v_i^2$ 
            adjacent to some $u \in N_{{G_P^{i-1}}'}(v_i)$ such that 
            $u \notin N_{{G_P^{i-1}}''}(v_i)$}
        \State Split $v_i \in V({G_P^{i-1}}'')$ into $v_i^1, v_i^2$.
        \State For each such $u$, make neither $v_i^1$ nor $v_i^2$ adjacent to $u$ in ${G_P^{i-1}}''$.
    \Else
        \State Proceed as in the inclusive split version.
    \EndIf
\EndFor
\State \textbf{return} ${G^{k_{\mathrm{in}}}}''_P$
\end{algorithmic}
\end{algorithm}

Following these changes and the same sequence as the inclusive vertex split (modeled as exclusive vertex split), we obtain a graph $G_P'' (={G^{k_{\mathrm{in}}}}''_P)$. The number of splits in both constructions is identical, i.e., both sequences contain $k_{in}$ splits. Furthermore, the resulting edge set of $G_P''$ is a subset of that of the inclusive split graph, i.e., $E(G_P'') \subseteq E(G_P')$
Consequently, $G_P''$ is a linear forest, possibly with some singleton components.

Any singleton vertex $v$ in $G_P''$ can be disregarded: if such a vertex arose from splitting a vertex $u$ into $v$ and $w$, then $v$ is isolated while $w$ plays the role of $u$. Omitting that split would reduce the total number of splits by one without altering the final structure, i.e, it would still consist solely of acyclic components of maximum degree 2. Therefore, we can assume without loss of generality that $G_P''$ contains no singleton vertices.

This shows that a linear forest can be constructed using at most $k_{in}$ exclusive vertex splits. This says $k_{in} \leq k_{ex}$, which contradicts the minimality of $k_{ex}$.
\end{proof}

To prove the main result of this section, we make use of the following well-known lemma, whose proofs can be found in ~\cite{west2001introduction, diestel2017graph}.  

\begin{lemma}\label{minnoofwalks}
    For a graph $G$, with $\alpha>1$ many vertices with odd degree, the minimum number of trails (denoted as $\textbf{MNT}$) that decompose it is exactly $\frac{\alpha}{2}$, i.e $\textbf{MNT}=\frac{\alpha}{2}$. However, if the connected graph has $\alpha=0,$ i.e, it is an Eulerian graph, a single trail (the Eulerian closed trail) can cover the entire edge set, i.e, $\textbf{MNT}=1$.   
    More succinctly, for a nontrivial connected graph with $\alpha$ odd-vertices, the minimum number of trails ($MNT$) that decompose it is $\max(\frac{\alpha}{2},1)$, i.e $\textbf{MNT}=\max(\frac{\alpha}{2},1)$.
\end{lemma} 

\begin{theorem}\label{Res:RLFEVS}
 In polynomial time, one can compute the minimum number of exclusive vertex splits required to construct a linear forest $G_P$ from an arbitrary graph $G$.   
\end{theorem}

\begin{proof}
    Similar to how we argued that the cycles after cycle graph vertex splitting corresponds to closed trails, one can argue that the paths in $G_{P}$ after linear forest exclusive vertex split correspond to trails in the graph $G$.

    Since we consider only \emph{exclusive vertex splits}, the number of edges remains invariant throughout the execution of the procedure; that is, $|E(G)| = |E(G_P)|$. Each vertex split, regardless of type, increases the number of vertices in the graph by exactly one. Hence, we have the following relationship:
\[
|V(G_P)| - \mathcal{C}(G_P) = |V(G)| + k - \mathcal{C}(G_P) = |E(G_P)|.
\]
Here, $\mathcal{C}(G)$ denotes the number of connected components of $G$.
In the case of exclusive vertex splits, the number of edges remains unchanged, i.e.
\[
|E(G_P)| = |E(G)|.
\]

Therefore, for exclusive vertex splits, we can express:
\[
k = |E(G)| - |V(G)| + \mathcal{C}(G_P).
\]
where $k$ denotes the number of vertex splits performed.
As the quantities $|E(G)|$, $|V(G)|$, and consequently $|E(G_P)|$ are fixed for a given graph $G$, it follows that the number of exclusive vertex splits depends solely on the number of connected components $\mathcal{C}(G_P)$ obtained after executing the splitting sequence. Minimizing the number of exclusive vertex splits, therefore, corresponds to minimizing the number of edge-disjoint trails required to cover all edges of $G$. First, we consider the case when $G$ is a non-Eulerian graph. By the handshake lemma, the number of odd-degree vertices is at least $2$. Hence, by \cref{minnoofwalks}, the minimum number of trails required to cover all edges is equal to $\frac{\alpha}{2}$.

Although we do not explicitly describe the process of “opening” these trails, the construction is analogous to the procedure outlined in \cref{openingwalk}. Specifically, we conceptually introduce $\frac{\alpha}{2}$ auxiliary (or \emph{false}) edges between paired odd-degree vertices, apply the same opening procedure as in \cref{openingwalk}, and subsequently remove these auxiliary edges. The resulting configuration yields exactly $\frac{\alpha}{2}$ edge-disjoint paths (i.e the graph $G_P$). We can get to these $\frac{\alpha}{2}$ by following the same procedure as with the false edges, but anytime we assign the false edge to one of the descendants of a vertex, we skip such a step during our vertex split to the linear forest ($G_P$). Thus, we achieve the minimum number of vertex exclusive vertex splits required is $k = |E(G)| - |V(G)| + \frac{\alpha}{2}$.

If $G$ is an Eulerian graph, then we open the trail using $|E(G)| - |V(G)|$ splits and then do a single exclusive split to open the cycle into a path (note that this many splits for an Eulerian graph is essential since we at least have one component in the linear forest, i.e $k\geq |E(G)| - |V(G)| + \mathcal{C}(G_P)=|E(G)| - |V(G)| + 1$ ). In general we have $|E(G)|-|V(G)|+\textbf{MNT}$ many splits. 
\end{proof}

\begin{corollary}\label{Res:RLFIVS}
    In polynomial time, one can compute the minimum number of Inclusive vertex splits required to construct a linear forest $G_P$ from an arbitrary graph $G$. 
\end{corollary}
\begin{proof}
    From Observation \ref{exisin} we can easily see the proof.
\end{proof}

\section{{\sc Bipartite Inclusive/Exclusive Vertex Splitting}}~\label{Sec: Bipartite}
In this section, we prove that {\sc Bipartite-Inclusive Vertex Splitting} and {\sc Bipartite-Exclusive Vertex Splitting} are NP-complete. Firbas and Sorge \cite{Splitthesis,FirbasSorge2024} showed that {\sc Bipartite-Inclusive Vertex Splitting} is NP-hard by providing a reduction from {\sc 2-Subdivided Cubic Vertex Cover}. We establish that {\sc Bipartite-Inclusive Vertex Splitting} is NP-complete by 
presenting a reduction from {\sc Bipartite-Vertex Deletion} (popularly known as {\sc Odd Cycle Transversal}). The advantage of our reduction is that it demonstrates the equivalence between  {\sc Bipartite-Inclusive Vertex Splitting} and {\sc Bipartite-Vertex Deletion}. So, the existence of a kernelization or an FPT algorithm for {\sc Bipartite-Vertex Deletion} implies the existence of the same for {\sc Bipartite-Inclusive Vertex Splitting}. Moreover, our proof is comparatively easy to understand.

We first define our problems formally.

\fbox
    {\begin{minipage}{38.7em}\label{BIVS}
       {\sc Bipartite-Inclusive Vertex Splitting (Bipartite-IVS)}\\
        \noindent{\bf Input:} An undirected graph $G=(V,E)$, and a positive integer $k$.\\
    \noindent{\bf Question:} Is there a sequence of at most $k$ inclusive vertex splitting that transforms $G$ into a bipartite graph?
    \end{minipage} }\\
\vspace{3mm}

\fbox
    {\begin{minipage}{38.7em}\label{BEVS}
       {\sc Bipartite-Exclusive Vertex Splitting (Bipartite-EVS)}\\
        \noindent{\bf Input:} An undirected graph $G=(V,E)$, and a positive integer $k$.\\
    \noindent{\bf Question:} Is there a sequence of at most $k$ exclusive vertex splitting that transforms $G$ into a bipartite graph?
    \end{minipage} }\\
\vspace{3mm}

\begin{theorem}
    Let $G$ be a graph. Then $(G,k)$ is a yes-instance of {\sc Bipartite-Inclusive Vertex Splitting} if and only if $(G,k)$ is a yes-instance of {\sc Bipartite-Vertex Deletion}. 
\end{theorem}
\proof For the forward direction, we assume that $(G,k)$ is a yes-instance of {\sc Bipartite-Inclusive Vertex Splitting}. Let $S$ be the set of vertices of $G$ that are split to obtain a bipartite graph, with $|S| \leq k$. Suppose that
applying inclusive vertex splitting to the vertices in $S$ transforms $G$ into a bipartite graph $G'$.

We claim that $G-S$ is a bipartite graph. Note that $G-S$ is an induced subgraph of $G'$ because splitting the vertices in $S$ does not affect the adjacency between the vertices in $G-S$. Since $G'$ is bipartite and $G-S$ is its induced subgraph, it follows that $G-S$ is also bipartite. So we conclude that $S$ is a bipartite vertex deletion set of size almost $k$. Hence, $(G,k)$ is a yes-instance of {\sc Bipartite-Vertex Deletion}.

For the reverse direction, we assume that $(G,k)$ is a yes-instance of {\sc Bipartite Vertex Deletion}. Suppose $S=\{v_1,v_2, \ldots, v_l\}$ is a bipartite-vertex deletion set of $G$ with $l \leq k$. Suppose $G-S$ is a bipartite graph with bipartition $V_1$ and $V_2$. We treat $S$ as an ordered set.

Now we construct a sequence $G=G_0,G_1,G_2, \ldots,G_l$ where $G_l$ is obtained by splitting of vertex $v_i$ in the graph $G_{i-1}$. Note that the construction of $G_i$ is done inductively. Now, we will see how the vertex splitting is done to get $G_i$ from $G_{i-1}$.

For $1\leq i\leq l$, we split $v_i$ into two vertices $v_{i1}$ and $v_{i2}$ such that 
$$N_{G_{i}}(v_{i1})= (N_{G_{i-1}}(v_i) \cap V_2 )  \cup \left\{v_{j}: v_j \in N-{G_{i-1}}(v_i) \text{ when } j>i \right\} $$
$$N_{G_i}(v_{i2})= (N_{G_{i-1}}(v_i) \cap V_1 )  \cup \{v_{j1}: v_{j1} \in N_{G_{i-1}}(v_i)\} $$

We see that $V_1 \cup \{v_{11},v_{21},\ldots v_{l1}\}$ and $V_2 \cup \{v_{12},v_{22},\ldots v_{l2}\}$ is a bipartition of the graph $G_l$. Hence, $(G,k)$ is a yes-instance of {\sc Bipartite-Inclusive Vertex Splitting}. \qed
\vspace{3mm}

It is easy to see that {\sc Bipartite-Inclusive Vertex Splitting} and  {\sc Bipartite-Exclusive Vertex Splitting} are in NP. Since we know that {\sc Bipartite Vertex Deletion} (popularly known as {\sc Odd Cycle Transversal}) is NP-complete, {\sc Bipartite-Inclusive Vertex Splitting} is NP-hard. We note that in the reverse direction of the proof, all the vertex splittings that were done are exclusive vertex splittings. There is no problem in the forward direction if we allow only exclusive vertex splitting. It is easy to conclude the following results from the above discussion.

\begin{corollary}\label{Res:BIVS}
    {\sc Bipartite-Inclusive Vertex Splitting} is NP-complete.
\end{corollary}

\begin{corollary}
    Let $G$ be a graph. Then $(G,k)$ is a yes-instance of {\sc Bipartite-Exclusive Vertex Splitting} if and only if $(G,k)$ is a yes-instance of {\sc Bipartite-Vertex Deletion}. 
\end{corollary}

\begin{corollary}\label{Res:BEVS}
    {\sc Bipartite-Exclusive Vertex Splitting} is NP-complete.
\end{corollary}

\section{Conclusion}\label{Sec: Conclusion}
We have studied the classical complexity of {\sc $\mathcal{F}$-Inclusive/Exclusive Vertex Splitting} for different classes $\mathcal{F}$-constellation, cycle graph, linear forest, and bipartite graph. We show that {\sc Constellation-Inclusive Vertex Splitting} and {\sc Constellation-Exclusive Vertex Splitting} are NP-complete on cubic graphs and planar graphs of maximum degree $3$. On the positive side, we show that both {\sc Cycle Graph-Inclusive/Exclusive Vertex Splitting} and {\sc Linear Forest-Inclusive/Exclusive Vertex} are polynomial-time solvable.

We prove the equivalence of {\sc Bipartite-Inclusive/Exclusive Vertex Splitting} and {\sc Bipartite Vertex Deletion}. This guarantees the existence of a kernel and FPT algorithm for {\sc Bipartite-Inclusive/Exclusive Vertex Splitting} as these exist for {\sc Bipartite Vertex Deletion}.

Here, we have studied the classical complexity of {\sc Constellation-Inclusive/Exclusive Vertex Splitting} on cubic graphs and planar graphs of maximum degree $3$. The complexity of this problem on other graph classes remains open. Additionally, the parameterized complexity of this problem remains open, and it is interesting to explore whether this problem admits a polynomial-size kernel. 

\bibliographystyle{plainurl}
\bibliography{References}
\end{document}